\documentclass{conm-p-l}
\usepackage{graphicx} 
\usepackage{amsmath}
\usepackage{amsfonts}
\numberwithin{equation}{subsection}
\usepackage{amsthm}
\usepackage{amssymb}
\usepackage{graphicx}
\usepackage{setspace}
\usepackage[dvipsnames]{xcolor}
\usepackage{enumitem}
\usepackage{float}
\usepackage{caption}
\usepackage{stmaryrd}
\usepackage{ragged2e}
\usepackage{comment}
\usepackage{longtable}
\usepackage{pdflscape}
\newtheorem*{theorem}{Main Theorem}

\newtheorem{prop}{Lemma}[subsection]

\usepackage[textwidth=0.89in,textsize=scriptsize]{todonotes}

\newcommand{\SO}{\mathrm{SO}}
\newcommand{\Co}{\mathrm{Co}}
\newcommand\Spin{\mathrm{Spin}}
\DeclareMathOperator{\tr}{tr}
\DeclareMathOperator\str{str}
\newcommand{\Mod}{\mathrm{Mod}}

\newcommand\rO{\mathrm O}
\newcommand\rH{\mathrm H}
\newcommand\bC{\mathbb C}
\newcommand\bQ{\mathbb Q}
\newcommand\bZ{\mathbb Z}
\newcommand\bR{\mathbb R}
\newcommand\cH{\mathcal H}
\newcommand\rR{\mathrm{R}}
\newcommand\NS{\mathrm{NS}}
\newcommand\bN{\mathbb N}
\newcommand\rU{\mathrm{U}}
\newcommand\cL{\mathcal L}

\DeclareMathOperator\Aut{Aut}

\usepackage[
backend=biber,
style=alphabetic,
sorting=ynt
]{biblatex}
\usepackage{soul} 

\newcommand\group{\mathcal S}

\title{Ground-state degeneracy of twisted sectors of Conway Moonshine SCFT}
\author{Alissa Furet and Theo Johnson-Freyd}
\thanks{This article summarizes the main results of the first author's MSc thesis (Dalhousie University, 2022). The second author thanks Ying Lin for related discussion. Both authors were supported by the NSERC Discovery Grant RGPIN-2021-02424 and by the Simons Collaboration on Global Categorical Symmetries. \\ \textsc{Dalhousie University, Department of Mathematics and Statistics} \\ \textsc{Perimeter Institute for Theoretical Physics} \\ \url{alissa.furet@mail.mcgill.ca}, \url{theojf@pitp.ca}}

\begin{document}

\maketitle

\begin{abstract}
  We calculate the ground state degeneracies of all twisted sectors in the ``Conway Moonshine'' holomorphic SCFT $V^{f\natural}$. We find that almost all sectors have ground states of only a single parity: specifically, 66 twisted sectors have nontrivial ground states of a single parity, 39 twisted sectors have spontaneous supersymmetry breaking, and only 6 twisted sectors have ground states of both parities. Although ``nontrivial ground states, all of the same parity'' is the expected behavior for a generic SQM model without symmetry protection, it is surprising in the presence of a large symmetry group, as is the case in $V^{f\natural}$. This surprise  hints that there are as-yet-undiscovered features of Conway Moonshine.
\end{abstract}

\section{Introduction}
In a supersymmetric quantum mechanics (SQM) model --- a super Hilbert space equipped with an odd self-adjoint operator $\hat{G}$ --- the bosonic and fermionic modes almost exactly pair up. More precisely, in any eigenspace for the Hamiltonian $\hat{H} \propto \hat{G}^2$ (the precise proportionality factor depends on conventions) with eigenvalue $E \neq 0$, the supersymmetry $\hat{G}$ supplies an isomorphism between the bosonic and fermionic subspaces. But in the ground state --- the eigenvectors with eigenvalue $E=0$ --- there are typically different numbers of bosons and fermions unless there is some extra symmetry identifying them. In the absence of any symmetry, one expects a generic SQM model to have all ground states of the same parity: a small deformation to $\hat{G}$ will typically annihilate ground states in boson-fermion pairs. For special models, there are no ground states at all: such models are said to have \emph{spontaneous supersymmetry breaking}.

The story is more subtle in the presence of a symmetry group $\group$. In this case, the boson and fermion ground states are each (typically reducible) representations of $\group$, and a small $\group$-symmetric deformation to $\hat{G}$ can only annihilate pairs of states in isomorphic representations. As such, a generic SQM model enhanced with $\group$-symmetry will typically have both boson and fermion ground states, of different numbers. We will refer to a symmetry-enhanced SQM model as \emph{conspiratorial} when it lacks spontaneous supersymmetry breaking but nevertheless conspires to have only ground states of a single parity.

One interesting construction of SQM models with nontrivial symmetry groups starts with a highly-symmetric holomorphic 1+1D superconformal field theory (SQFT) $V$. Pick any supersymmetry-preserving automorphism $g$ of $V$. Then the $g$-twisted Ramond sector of $V$, which we will denote  $V_{\rR g}$, comes with a supersymmetry generated by the zero mode of the supercurrent, and comes with a projective symmetry by the centralizer $C(g)$ of $g \in \operatorname{Aut}(V)$. Fairly often in these SQM models there are  equal numbers of bosonic and fermionic ground states. For example, this happens whenever $g$ commutes with a continuous family of supersymmetry-preserving automorphisms of $V$ (see Lemma~\ref{prof:fermionzeromode}). The most important holomorphic SCFT without continuous automorphisms is  the ``Conway Moonshine'' theory $V^{f\natural}$ of J.\ Duncan \cite{DunSMforConLargGrp}. It is the unique $c\leq 12$ holomorphic SCFT without continuous automorphisms, and its (finite) group of automorphisms is a copy of Conway's largest sporadic simple group $\Co_1$.

We decided to calculate the number of ground states of each parity for each of the $g$-twisted Ramond sectors of $V^{f\natural}$. We find the following:

\begin{theorem}\label{mainresult}
  For exactly 66 of the 101 conjugacy classes $[g]$ in $\Co_1$, the $g$-twisted Ramond sector $V^{f\natural}_{\rR g}$ of $V^{f\natural}$ is conspiratorial: it has nontrivial ground states of a single parity. For exactly 39 conjugacy classes, $V^{f\natural}_{\rR g}$  has spontaneous supersymmetry breaking. Exactly six conjugacy classes are nonconspiratorial, with ground states in $V^{f\natural}_{\rR g}$ of both parities.
\end{theorem}

The existence of six nonconspiratorial conjugacy classes means that there cannot be an easy universal explanation for why so many conjugacy classes in $\Co_1$ are conspiratorial. Rather, we believe that explaining our result will shed new light on Conway moonshine. In lieu of an explanation, we offer simply an observation.
The six nonconspiratorial classes are (in the standard ATLAS naming convention):
$$ 3C, \quad 4B, \quad 5C, \quad 6F, \quad 8D, \quad 9B.$$
These are some (but not all) of the small number of conjugacy classes in $\Co_1$ whose Frame shapes are unbalanced (compare \cite{Theo}). The latter three classes are some (but not all) of the square and cube roots of the first two classes: $6F^2 = 3C$, $8D^2 = 4B$, and $9B^3 = 3C$.
Moreover, classes $3C$ and $5C$, together with class $2A$ (which is conspiratorial), are particularly special among conjugacy classes in $\Co_1$:  these are the only classes with centralizers of shape $(\text{extrapsecial}).(\text{almost simple})$; and these are the only classes $g$ of prime order for which the (twisted) orbifold $V^{f\natural} \sslash \langle g\rangle$ is the supersymmetric $E_8$ lattice VOA used in the original construction of $V^{f\natural}$ in \cite{duncanPHD}.

The structure of this article is as follows. In Section~\ref{sec:preliminaries}, we briefly review some definitions and explain why $g$-twisted Ramond sectors of SCFTs are naturally SQMs. We moreover explain why typically these SQMs are far from conspiratorial: not only are there usually both bosonic and fermionic ground states, but the same number of them. In Section~\ref{sec:methods} we outline the methods we used to do the  computations summarized in our Main Theorem. Finally, Section~\ref{sec:results} tabulates the detailed outputs of those computations. Our computer code can be found in the Ancillary Files section of this article's arXiv posting.

\section{Preliminaries}\label{sec:preliminaries}
In this section, we briefly review some definitions and explain why $g$-twisted Ramond sectors of SCFTs are naturally SQMs.

\subsection{SQM models and Witten index}

A \emph{quantum mechanics} (QM) model consists of a complex Hilbert space $\cH$ equipped with a self-adjoint operator $\hat{H}$ called the \emph{Hamiltonian}. 
We will abusively write simply $\cH$ for the QM model $(\cH,\hat{H})$.
The Hamiltonian is typically unbounded, but we require that its spectrum is bounded below. In this case, the \emph{Euclidean time-evolution operator}
$$ U_s = e^{-s\hat{H}}, \quad s \in \mathbb{R}_{\geq 0}$$
is bounded for all $s \geq 0$. The QM model $\cH$ is called \emph{compact} if $U_s$ is moreover trace-class for all $s > 0$. The \emph{partition function} of a compact QM model is
$$Z(\cH) : s \mapsto \tr_{\mathcal H} \bigl(U_s\bigr), \quad s > 0. $$
 Note that it is common to work with an exponentiated coordinate $q:= e^{-s}$, in which case compactness of $\cH$ says that $Z(\cH)(q)$ converges for $|q|<1$.
 
A \emph{fermionic} QM model is just as above, but instead of an ordinary Hilbert space, $\mathcal{H} = \mathcal{H}_{\bar 0} \oplus \mathcal{H}_{\bar 1}$ is a super (aka $\mathbb Z_2$-graded) Hilbert space. The operator which is $+1$ on the even, aka bosonic, summand $\mathcal{H}_{\bar 0}$ and $-1$ on the odd, aka fermionic, summand $\mathcal{H}_{\bar 1}$ is called $(-1)^f$. The Hamiltonian in a fermionic QM model is required to be even, i.e.\ to commute with $(-1)^f$. A compact fermionic QM model has two partition functions:
\begin{gather}
\label{eqn:ZNS}  Z^{\NS}(\cH) : s \mapsto \tr_{\mathcal  H} \bigl(U_s\bigr), \\
\label{eqn:ZR}  Z^{\rR}(\cH) : s \mapsto \tr_{\mathcal H} \bigl((-1)^f U_s\bigr).
\end{gather}
The first is called the \emph{Neveu--Schwarz partition function} and the second is called the \emph{Ramond partition function}. Note that this trace is the trace on the underlying non-super Hilbert space of $\mathcal H$; the operation $X \mapsto \tr \bigl((-1)^f X\bigr) $ is often referred to as the \emph{supertrace} $\str(X)$. One could more generally pick any automorphism $g \in \Aut(\cH)$, and study its character
$$ Z^g(\cH) : s \mapsto \tr_{\mathcal  H} \bigl(gU_s\bigr),$$
(in which case, to match notation, we should simply write $\rR$ for $(-1)^f$).

As mentioned already in the introduction, \emph{supersymmetric} QM (aka SQM) models are equipped with a self-adjoint \emph{odd} operator $\hat G$ whose square is proportional (by a fixed convention-dependent constant) to $\hat{H}$. More precisely, if there is one such chosen $\hat G$, then the model is said to have \emph{$N=1$ supersymmetry}; we will never use more than one supersymmetry, and so we will not write ``$N=1$'' in the sequel.
The following result is well-known:

\begin{prop}[\cite{witten}] \label{prop:constant}
    In any compact SQM model $\cH$, the Ramond partition function $Z^{\rR}(\cH)$ is an integer (and independent of $s$).
\end{prop}

The integer $Z^{\rR}$ is called the \emph{Witten index} of the SQM model.

\begin{proof}  Let $\cH^{\lambda} \subset \cH$ denote the eigenspace with eigenvalue $\hat{H} = \lambda$.  Since $\mathcal{H}$ is a super Hilbert space and $\hat{H}$ preserves the super structure, $\mathcal{H}^{\lambda}$ is also a super space, and decomposes as a sum of bosonic and fermionic pieces, $\mathcal{H}^{\lambda}= \mathcal{H}^{\lambda}_{\bar 0} \oplus \mathcal{H}^{\lambda}_{\bar 1}$. The partition function is then  
$$ Z(\cH)(q) = \sum_{\lambda}q^\lambda (\dim(\mathcal{H}^{\lambda}_{\bar 0}) \pm \dim(\mathcal{H}^{\lambda}_{\bar 1})),$$ where addition is Neveu-Schwarz and subtraction is Ramond. 

Since $\hat{G}$ commutes with $\hat{H}$, it preserves the eigenspaces; since it is odd, it maps $\hat{G}|_{\cH^\lambda} : \mathcal{H}^{\lambda}_{\bar 0} \leftrightarrows \mathcal{H}^{\lambda}_{\bar 1}$. The composition of these two maps is $\hat{H} = \lambda$. In particular, if $\lambda \neq 0$ then $\hat{G}^2$ is an isomorphism so $\hat{G}$ is also an isomorphism. In particular $\dim(\mathcal{H}^{\lambda}_{\bar 0}) = \dim(\mathcal{H}^{\lambda}_{\bar 1})$ when $\lambda \neq 0$, and so
$$ Z^{\rR}(\mathcal{H}) = q^0 \bigl(\dim(\mathcal{H}^{\lambda}_{\bar 0}) - \dim(\mathcal{H}^{\lambda}_{\bar 1})\bigr),$$
which is an $s$-independent integer.
\end{proof}

If $\hat{G}$ exists, then all eigenvalues of $\hat{H}$ are non-negative and so
$Z^{\NS}(\cH)(q)$ is  smooth at $q=0$, equivalently $s = +\infty$. Moreover, its constant term  $Z^{\NS}(\cH)(q=0)$ counts (without sign) the number of states of eigenvalue $0$. Thus we can retrieve the number of boson and fermion ground states via: 
\begin{gather}
\label{eqn:countB}    \#\text{bosons} = \frac{Z^{\NS}(\cH)(q=0) + Z^{\rR}(\cH)}{2}, \\
\label{eqn:countF}    \#\text{fermions} = \frac{Z^{\NS}(\cH)(q=0) - Z^{\rR}(\cH)}{2}.
\end{gather}

\subsection{Holomorphic SCFTs}

In order to produce interesting SQM models, we will work with a specific holomorphic superconformal field theory $V^{f\natural}$ first discovered by J.\ Duncan \cite{duncanPHD,DunSMforConLargGrp}. We will not review the precise definition of ``holomorphic conformal field theory.'' Suffice it to say that a \emph{holomorphic CFT} is a (full) conformal field theory in which all vertex operators have only holomorphic dependence on the location of insertion. Holomorphic CFTs have a nice axiomatization in terms of vertex operator algebras (VOAs): they are precisely the (unitary and of CFT-type) VOAs $V$ such that the category $\Mod(V)$ of vertex modules is semisimple with a unique simple object (namely $V$ itself). We will without further comment allow our VOAs to have both bosonic and fermionic elements. Unitarity enforces the spin-statistics relation, which says that bosonic vertex operators have (non-negative) integral conformal weights whereas fermionic vertex operators have conformal weights in $\bN + \frac12$.

Suppose given a VOA $V$ and a finite-order automorphism $g \in \Aut(V)$. Then there is a well-studied category $\Mod_g(V)$ of \emph{$g$-twisted modules} (we recommend the textbook \cite{IntrotoVAandTheirRep}). Holomorphicity of $V$ implies that all of the categories $\Mod_g(V)$ are semisimple with a unique simple module. This unique simple module is called the \emph{$g$-twisted sector} of $V$, and we will denote it $V_g$. Note that the words ``the'' and ``unique'' are slightly abusive: $V_g$ is well-defined up to non-canonical, possibly-odd isomorphism. (In particular, the overall parity of $V_g$ is not well-defined.)

A special case is when $g = (-1)^f$ is the canonical parity operator on $V$. The $(-1)^f$-twisted sector is then conventionally called the \emph{Ramond sector} $V_\rR$ of $V$. (Thus, to match notation, we should write simply $\rR$ for $(-1)^f$.) The untwisted ($g=1$) sector is conventionally called the \emph{Neveu--Schwarz sector} $V = V_{\NS}$. More generally, given $g \in \Aut(V)$, we will refer to $V_g = V_{\NS g}$ as the \emph{$g$-twisted Neveu--Schwarz sector} and $V_{(-1)^fg} = V_{\rR g}$ as the \emph{$g$-twisted Ramond sector}.

The \emph{zero-mode} of a bosonic vertex operator $X(-)$ is the operator $X_0 : V \to V$ defined by
\begin{equation}\label{eqn:zeromode}
X_0[Y] := \oint_{z \in U(1)} \frac{1}{2 \pi i z}X(z) Y(0) \mathrm{d} z.
\end{equation}
 This formula makes sense on the $g$-twisted sector $V_g$ provided that $g$ preserves $X$. If $X$ is fermionic, then $X(z)$ is anti-periodic if $z$ goes around a circle with even spin structure, but periodic if $z$ goes around a circle with odd spin structure; the upshot is that $X_0$ is defined on $V_\rR$, and on $V_{\rR g}$ if $g$ preserves $X$, but not on $V_{\NS}$. 
 
 A special case is when $X = L$ is the conformal vector, which (by definition) is preserved by all symmetries $g \in \Aut(V)$. Its zero mode $L_0$ is then a reasonable choice of Hamiltonian for a quantum mechanics model with Hilbert space $V_{g}$. Actually, although it induces the correct time evolution on the projective space of states, $L_0$ is not quite the Hamiltonian usually used. Rather, the Hamiltonian on $V_g$ is always taken to be $\hat{H} := L_0 - \frac{c}{24}$, where $c$ is the central charge of the CFT $V$. The  shift by the constant $\frac{c}{24}$ can be explained by carefully computing how the physics transforms when changing variables between rectilinear and polar coordinates (the buzzword is ``Schwarzian derivative''); we refer the reader to \cite{YellowBook} for further discussion.

Applying (\ref{eqn:ZNS}--\ref{eqn:ZR}) to the Neveu--Schwarz and Ramond sectors supplies partition functions
\begin{align*}
  Z^{\NS}_{\NS}(V) &= \tr_{V_\NS}\bigl( q^{L_0 - c/24}\bigr), &
  Z^{\rR}_{\NS}(V) &= \str_{V_\NS}\bigl( q^{L_0 - c/24}\bigr), \\
  Z^{\NS}_{\rR}(V) &= \tr_{V_\rR}\bigl( q^{L_0 - c/24}\bigr), &
  Z^{\rR}_{\rR}(V) &= \str_{V_\rR}\bigl( q^{L_0 - c/24}\bigr).
\end{align*}
More generally, we can use any twisting. We will be most interested in the $g$-twisted Ramond sectors:
$$ Z^{\NS \text{ or }\rR}_{\rR g}(V) = Z^{\NS \text{ or }\rR}(V_{\rR g}) = (\tr \text{ or } \str)_{V_{\rR g}}\bigl( q^{L_0 - c/24}\bigr).$$
We remind that $V_{\rR g}$ is well-defined only up to a possibly-odd isomorphism. This does not affect the value of $Z^\NS_{\rR g}$, but it means that $Z^{\rR}_{\rR g}$ is well-defined only up to a sign. 

In Section~\ref{sec:Frame}, we will use more generally the \emph{twisted and twined} Ramond--Ramond partition function, defined as follows. In general, for any $g,h \in \Aut(V)$, the action of $v \in V$  on $V_g$ via $v \mapsto h(v)$ instead of the usual action writes $V_g$ as a simple $hgh^{-1}$-twisted $V$-module. Thus $h$ supplies an isomorphism $V_g \cong V_{hgh^{-1}}$, well defined up to a $\rU(1)$-phase ambiguity. Now suppose that $g$ and $h$ commute. Then $V_{hgh^{-1}}$ is literally equal to $V_g$, and so $h$ supplies an automorphism of $V_g$, again well defined up to a $\rU(1)$-phase ambiguity. The character of this automorphism is the twisted and twined partition function:
\begin{equation}
\label{eqn:twined} Z^{\rR h}_{\rR g}(V) = \str_{V_{\rR g}}\bigl( h q^{L_0 - c/24}\bigr) = \tr_{V_{\rR g}}\bigl((-1)^f h q^{L_0 - c/24}\bigr).\end{equation}

A \emph{supersymmetry} in a CFT of central charge $c$ is a vertex  operator $\hat{G}(-)$ of conformal weight $\frac32$ solving the OPE
\begin{equation} \label{eqn:scft} \hat{G}(z)\hat{G}(0) \sim \frac{\frac23 c}{z^3} + \frac{2L(0)}z\end{equation}
where $L(-)$ is the stress-energy vertex operator. Equation 
 \eqref{eqn:scft} implies that $\hat{G}_0^2 \propto L_0 - \frac{c}{24} = \hat{H}$. In other words, for each automorphism $g$ of the SCFT $(V, \hat{G})$, the $g$-twisted Ramond sector $V_{\rR g}$ is naturally an SQM model. Note that the supersymmetry is odd, and so \emph{not} preserved by $(-1)^f$. In particular, $V_{\rR g} \neq V_{\NS h}$ for any supersymmetry-preserving $h$; indeed, the latter is not naturally an SQM model.

Quite often the SQM model $V_{\rR g}$ has vanishing Witten index:

\begin{prop}\label{prof:fermionzeromode}
  Suppose that $V$ is a holomorphic SCFT and that $g \in \Aut(V)$ commutes with some the continuous symmetry. Then the $g$-twisted Ramond sector $V_{\rR g}$ has an equal number of bosonic and fermionic ground states.
\end{prop}

\begin{proof}
  Noether's theorem for holomorphic SCFTs identifies the continuous symmetries of $V$ with the vertex operators of spin $1/2$. The way this works is as follows. Noether's theorem for non-super holomorphic CFTs is well-known: the continuous symmetries are precisely the operators $X_0$ for vertex operators $X(-)$ of spin $1$. Now suppose that $X_0$ preserves the supersymmetry $\hat{G}$. In other words, $X_0[\hat{G}] = 0$, so that the full OPE of $X$ with $\hat{G}$ is
  $$ X(z) \hat{G}(0) \sim  \frac{Y(0)}{z^2} + \frac{0}{z}$$
  for some vertex operator $Y$ of spin $1/2$. Moreover, \eqref{eqn:scft} implies that 
  \begin{equation} \label{eqn:susydescent} \hat{G}(z) Y(0) \sim \frac{2X(0)}z.\end{equation}
   Conversely, any spin $1/2$ vertex operator $Y$ selects a spin $1$ operator $X$ via \eqref{eqn:susydescent}, and the corresponding continuous symmetry fixes $\hat{G}$.

  Now suppose that the continuous symmetry generated by $Y$ commutes with $g \in \Aut(V)$. Then $Y_0$ is defined on  $V_{\rR g}$.  
  Because $Y$ has spin $1/2$, its self-OPE with its complex conjugate $\bar{Y}$ is
  $$ Y(z) \bar{Y}(0) = \frac\lambda z$$
  for some (positive, hence nonzero) constant $\lambda = \|Y\|^2$. But then $Y_0 \bar{Y}_0 = \lambda \neq 0$, and since $Y$ is odd, $Y_0$ provides an isomorphism between the bosonic and fermionic states, of any energy including $0$, in $V_{\rR g}$.
\end{proof}


\section{Calculations} \label{sec:methods}

We turn now to calculating the partition functions $Z^\NS_{\rR g}(V^{f\natural}) = Z^\NS(V^{f\natural}_{\rR g})$ and $Z^\rR_{\rR g}(V^{f\natural}) = Z^\rR(V^{f\natural}_{\rR g})$ of the $g$-twisted Ramond sector of $V^{f\natural}$, as these will supply a count of the ground states via (\ref{eqn:countB}--\ref{eqn:countF}). 

If one forgets the supersymmetry, then the automorphism group of $V^{f\natural}$ is the group $\SO^+_{24}$, which is by definition the image of $\Spin_{24}$ in the positive half-spin representation. This is a group with the same Lie algebra as $\SO_{24}$, but a different global form: it is  the unique, up to the outer automorphism, double cover of $\mathrm{PSO}_{24}$ other than $\SO_{24}$; and it is the unique, up to the outer automorphism, $\bZ_2$-quotient of $\Spin_{24}$ other than $\SO_{24}$. In particular, twisted (and twined) partition functions of $V^{f\natural}$ make sense for all finite-order elements of $\SO^+_{24}$. We will explain a little bit more about how $\SO^+_{24}$ arises in Section~\ref{sec:thetaseries}. Although the automorphism group of $V^{f\natural}$ is just $\SO^+_{24}$, it is convenient to describe elements in terms of (either of their) lift(s) to $\Spin_{24}$.

The main theorem of \cite{DunSMforConLargGrp} states that $V^{f\natural}$ admits a unique-up-to-automorphism supersymmetry, and that the subgroup of $\SO^+_{24}$ that preserves this supersymmetry is a copy of Conway's largest sporadic group $\Co_1$. To help orient the reader, we note that the double cover $\Spin_{24} \to \SO^+_{24}$ restricts over $\Co_1$ to the double cover $\Co_0 \to \Co_1$, and that $\Co_0 \subset \Spin_{24}$ maps injectively along $\Spin_{24} \to \SO_{24} \subset \rO_{24}$. As a subgroup of $\rO_{24}$, the group $\Co_0$ arises as the group of automorphisms of the Leech lattice $\Lambda_{24} \subset \bR^{24}$.

\subsection{Frame shapes and conjugacy classes} \label{sec:Frame}

Frame shapes (introduced in \cite{frameog}) generalize cycle structures of permutations, where one writes a formal product like $1^2 2^1 3^1$ to indicate the conjugacy class of a permutation $(1)(2)(34)(567) \in S_7$. Frame shapes are defined as follows. Suppose that $g \in \rO_n$ has finite order $o(g) < \infty$ and that there exists a basis in which the matrix for $g$ has all entries in $\bQ$ --- in examples, this happens because $g$ preserves a lattice $\bZ^n \hookrightarrow \bR^n$. Then the characteristic polynomial $p(x) = \det(x - g)$ has coefficients in $\bQ$. On the other hand, because $g$ has finite order, its eigenvalues are all roots of unity. It follows that the irreducible-over-$\bQ$ factors of $p(x)$ are cyclotomic polynomials. Equivalent $p(x)$ factors (uniquely!) into a product of shape
$$ p(x) = \prod_{d | o(g)} (1-x^d)^{k_d},$$
where some of the exponents $k_d$ may be negative. We remark for future reference that the $x^1$ term in this product is
\begin{equation}
\label{eqn:trace} \tr_{\bR^n}(g) = k_1.
\end{equation}
The \emph{Frame shape} of $g$ is then defined to be the formal symbol
$$ \pi_g:=\prod_{d| o(g)} d^{k_d}.$$
Because Frame shapes exactly encode characteristic polynomials, they are a convenient and compact way of recording $\rO_n$-conjugacy classes.
In~\cite{kondo}, the Frame shapes for elements of $\Co_0 \subset \rO_{24}$ were calculated, resulting in 160 Frame shapes out of the 167 conjugacy classes.

We will need to modify this encoding in two ways. First, our formulas do not directly reference the Frame shape of $g$, but rather ($1/2\pi i$ times) the logarithms of the eigenvalues of $g$.
These log-eigenvalues can be quickly extracted from the Frame shape as follows. A general element of $\SO_{24}$ can be ``diagonalized'' to a block-diagonal matrix of shape
\[
    \left[
    \begin{array}{ccc}
     \left[
    \begin{array}{cc}  
    \cos(2\pi \lambda_1) &\sin(2\pi \lambda_1) \\ -\sin(2\pi \lambda_1) & \cos(2\pi \lambda_1)
    \end{array}
    \right]                                  \\& \ddots \\ &   &  \left[
    \begin{array}{cc}  
    \cos(2\pi \lambda_{12}) & \sin(2\pi \lambda_{12}) \\ -\sin(2\pi \lambda_{12}) & \cos(2\pi \lambda_{12})
    \end{array}
    \right]   
    \end{array}
    \right],
\] where each block is a rotation matrix. The $\rO_{24}$-conjugacy class is determined by the total log-eigenvalue $\vec\lambda \in (\bR/\bZ)^{12}$, up to permuations of the entries and sign-flips $\lambda _i \mapsto -\lambda_i \pmod 1$. The full list of log-eigenvalues $(\lambda_1,-\lambda_1,\lambda_2,\dots,-\lambda_{12})$ consists simply of the logarithms of the roots of the characteristic polynomial $\prod_{d | o(g)} (1-x^d)^{k_d}$. Each factor $d^{k_d}$ in the Frame shape contributes $k_d$-many copies of each of the log-eigenvalues $0,\frac 1 d, \frac 2 d, \dots, \frac{d-1}d \mod 1$ to the full list. Naturally, negative exponents entries in the Frame shape contribute negatively many  copies of the log-eigenvalues. After taking producing the full list of $24$ numbers, we then record just $12$ numbers by taking, for each $i$, whichever of $\lambda_i$ is in the interval $[0,\frac12] \subset \bR/\bZ$.

Second, given an element $g\in \Co_0$, our formulas depend not just on the $\rO_{24}$-conjugacy class of $g$, but on its conjugacy class in $\Spin_{24}$ --- because $\rH^1(\Co_0; \bZ_2) = \rH^2(\Co_0; \bZ_2) = 0$, the inclusion $\Co_0 \subset \rO_{24}$ lifts uniquely to an inclusion $\Co_0 \subset \Spin_{24}$. The fusion of conjugacy classes along $\Spin_{24} \to \rO_{24}$ is typically four-to-one:
\begin{enumerate}
    \item Whereas the $\rO_{24}$-conjugacy class of $\vec\lambda = (\lambda_1,\dots,\lambda_{12})$ is unchanged by any sign flip $\lambda_i \mapsto -\lambda_i$, the $\SO_{24}$-conjugacy class is unchanged only by an even number of sign flips. In other words, an $\rO_{24}$-conjugacy class parameterized by a vector $(\lambda_1,\dots,\lambda_{11},\lambda_{12})$ could have come from the $\SO_{24}$-conjugacy class parameterized by $(\lambda_1,\dots,\lambda_{11},\lambda_{12})$ or from the conjugacy class parameterized by $(\lambda_1,\dots,\lambda_{11},-\lambda_{12})$.
  \item Whereas the $\SO_{24}$-conjugacy class depends on the vector $\vec\lambda \in (\bR/\bZ)^{12} = \bR^{12}/\bZ^{12}$, the $\Spin_{24}$-conjugacy class depends on a vector in $\bR^{12} / D_{12}$, where $D_{12} \subset \bZ^{12}$ is the coweight lattice of $\Spin_{12}$, defined by
  $$ D_{12} = \{ \vec z \in \bZ^{12} : \sum z_i \in 2\bZ \}.$$
  In other words, an $\SO_{24}$-conjugacy class  $(\lambda_1,\dots,\lambda_{11},\lambda_{12})$ could have come from the $\Spin_{24}$-conjugacy class  $(\lambda_1,\dots,\lambda_{11},\lambda_{12})$ or from $(\lambda_1,\dots,\lambda_{11},1+\lambda_{12})$.
\end{enumerate}

One reason Frame shapes, equivalently $\rO_{n}$-conjugacy classes, are convenient to work with is that they are comparatively quick to compute from character tables: to work out the characteristic polynomial of a representation requires knowing only the power operations, some elementary symmetric polynomial theory, and a desktop computer to actually do the computations. (That said, we did not do this computation directly; Kondo already did that work in \cite{kondo}, and we simply extracted the log-eigenvalues from the Frame shapes.) For elements of $\Co_0$, it is in principle possible to in fact work out the $\Spin_{24}$-conjugacy classes from a careful analysis of the characters of the spin representations of $\Co_0$, but we did not pursue this route. Rather, we used a guess-and-check method guided by the fact that, for $V^{f\natural}$, we know that $Z^\rR_{\rR g}$ is an integer, and moreover we know which integer it is:
\begin{prop} \label{framshapeZrr}
Suppose that $\pi_g = \prod_d d^{k_d}$ is the Frame shape of any lift $\tilde g \in \Co_0$ of an element $g \in \Co_1$. Then $Z^\rR_{\rR g}(V^{f\natural}) = \pm k_1$. 
\end{prop}

We remind that $k_1 = \tr_{\bR^{24}}(\tilde g)$ by \eqref{eqn:trace}; that the two lifts $\tilde g$ of $g$, and hence their traces, differ only by a sign; and that $Z^\rR_{\rR g}(V^{f\natural})$ is only defined up to a sign.

\begin{proof} 
Zhu's Theorem states that if $V$ is any rational VOA and $M$ is any $V$-module, then $Z(M)$ is a scalar-valued modular function for some subgroup of the modular group $\Gamma = \mathrm{SL}_2(\mathbb{Z})$~\cite{zhu}, and moreover that the various $Z(M)$ are together a vector-valued modular form for the full modular group $\Gamma$. The main theorem by Carnahan and Miyamoto in~\cite{carnahanmiyamoto} implies that the twisted and twined characters (for finite-order automorphisms) of any rational VOA also participate in vector-valued modular forms. 

A special case says that, for any commuting elements $g,h$ and any holomorphic VOA $V$, the $S$-transformation $S = \bigl(\begin{smallmatrix} 0 & 1 \\ -1 & 0 \end{smallmatrix})$, which acts by $\tau \mapsto -1/\tau$ for $q = e^{-i\tau}$, takes
$$ Z^{\rR h}_{\rR g}(V) \overset S \longmapsto  Z^{\rR g}_{\rR h}(V)$$
up to an overall $\rU(1)$-phase ambiguity; see \eqref{eqn:twined} and the preceding discussion.

Specializing $V = V^{f\natural}$, and hence $c = 12$, and setting $h=1$, we find that $S$ takes
$$ Z^\rR_{\rR g}(V^{f\natural}) \overset S \longmapsto Z^{\rR g}_\rR (V^{f\natural})= \tr_{V^{f\natural}_\rR}\bigl( (-1)^f g q^{L_0 - c/24}\bigr).$$
But $Z^\rR_{\rR g}(V^{f\natural})$ is an integer by Lemma~\ref{prop:constant}, and so is unchanged under $S$-transformations. On the other hand, 
$$V_\rR^{f\natural} = \bR^{24} \oplus (\cdots)$$
where $L_0$ acts on $\bR^{24}$ with eigenvalue $\frac12$ and on $(\cdots)$ with higher eigenvalues, and where $\bR^{24}$ is all of one parity. Thus
$$ \tr_{V^{f\natural}_\rR}\bigl( (-1)^f g q^{L_0 - c/24}\bigr) = \tr_{\bR^{24}}(g) + O(q).$$
The result follows.
\end{proof}

All together, the way our algorithm works is that it takes the Frame shapes as calculated by \cite{kondo}, and first extracts the vector $\vec\lambda$ that parameterizes the $\rO_{24}$-conjugacy class. It then tries all four possible variations $\lambda_{12} \mapsto \lambda_{12},-\lambda_{12}, 1+\lambda_{12},1-\lambda_{12} \pmod 2$ in the formula for $Z^\rR_{\rR g}$. Remarkably, we find in each case that only one option produces the correct integer value. Our algorithm then uses this correct option in the formula for $Z^\NS_{\rR g}$. We turn now to describing those formulas.

\subsection{Lattices and $\theta$-series} \label{sec:thetaseries}


There are many constructions of $V^{f\natural}$. The most direct, and the most convenient for our purposes, is to recognize $V^{f\natural}$ as a lattice VOA for the $D_{12}^+$ lattice. By definition, $D_{12}^+$ is the weight (equivalently coweight, as it is self-dual) lattice of $\SO^+_{24}$. It is the extension of $D_{12} = \{\vec z \in \bZ^{12} : \sum_i z_i \in 2\bZ\}$ by the coset through the vector $\vec\varsigma := (\frac12, \frac12, \dots,\frac12)$. The construction that inputs a lattice $\cL$ and produces a VOA $V_\cL$ is well-known, and surveyed for example in \cite{VA}. In general, $\Aut(V_\cL)$ is the Lie group with weight lattice $\cL$. This is one way to see the equality $\Aut(V^{f\natural}) = \SO^+_{24}$.

For our purposes, we don't care about the precise vertex algebra structure on $V_\cL$, but we do care about the grading by $L_0$, so we briefly review the basics. Let $\cL \subset \bR^r$ be a full-rank integral lattice. There is a vertex algebra $\operatorname{Bos}(r) = \operatorname{Bos}(1)^{\otimes r}$ that depends only on $\bR^r$. As a vector space graded by $L_0$, there is a PBW-type isomorphism
$$ \operatorname{Bos}(1) \cong \bC[\phi,\partial \phi, \partial^2\phi, \dots] \cong \bigotimes_{n=1}^\infty \bC[\partial^n\phi]$$
where ``$\partial^n\phi$'' is supposed to be understood as a single symbol, denoting the $n$th generator of this infinitely-generated polynomial ring. The $L_0$-grading is compatible with the polynomial ring structure, and $\partial^n\phi$ has grading $n+1$. The VOA $\operatorname{Bos}(r)$ has central charge $c=r$ and is naturally a (non-rational) sub-VOA of $V_\cL$. Finally, there is also a PBW-type vector space isomorphism
$$ V_\cL \cong \operatorname{Bos}(r) \otimes \bC[\cL] = \bigoplus_{\vec\ell \in \cL} \operatorname{Bos}(r)$$
where the component graded by $\vec\ell$ picks up an extra $L_0$-grading by $\frac{\|\vec\ell\|^2}2$.

The twisted sectors are just as easy to write down. The maximal torus of $\Aut(V_\cL)$ is $\bR^r / \cL$. Assuming $g \in \Aut(V_\cL)$ is in the identity component, it can be conjugated into the maximal torus, and so is parameterized by its log-eigenvalues $\vec\lambda \in \bR^r/\cL$. The corresponding twisted sector $(V_\cL)_g$ is then PBW-isomorphic to
$$ (V_\cL)_g \cong \operatorname{Bos}(r) \otimes \bC[\cL+\vec\lambda] = \bigoplus_{\vec\ell \in \cL+\vec\lambda} \operatorname{Bos}(r)$$
where of course $\cL+\lambda$ just denotes the coset of $\cL$ through $\vec\lambda$, and where the $\ell$'th-component receives a contribution to its $L_0$-grading by $\frac{\|\vec\ell\|^2}2$.

These tensor products are compatible with the $L_0$-grading, and so the partition functions factor:
\begin{align}
\label{eqn:latticepartition} Z^{\NS \text{ or }\rR}_g(V_\cL) & = (\tr\text{ or }\str)_{(V_\cL)_g}\bigl(q^{L_0 - c}\bigr) \\
\notag & = (\tr\text{ or }\str)_{\operatorname{Bos}(1)}\bigl(q^{L_0 - c}\bigr)^r \times (\tr\text{ or }\str)_{\bC[\cL + \vec\lambda]}\bigl(q^{L_0}\bigr)
\\ \notag & = Z(\operatorname{Bos}(1))^r \times \sum_{\vec\ell \in \cL + \vec\lambda} \pm q^{\|\vec\ell\|^2/2}.
\end{align}
Note that $\operatorname{Bos}(1)$ is purely bosonic and so the choice of trace or supertrace is irrelevant. The sign in the sum is all $+1$ in the Neveu--Schwarz case, and $(-1)^{\|\vec\ell\|^2}$ in Ramond case. Actually, different vectors $\vec\lambda$ representing the same class in $\bR^r/\cL$ can lead to different signs, but the difference is uniform in $\ell$: it is just an overall sign-on $Z^\rR$, and in any case, $Z^\rR$ is typically only well-defined up to an overall sign.

It remains to compute $Z(\operatorname{Bos}(1))$ and the sum. The former is straightforward. There is an overall factor of $q^{-1/24}$ coming from the shift by $c/24$. The rest is the graded dimension of a polynomial algebra, or more precisely an infinite product $\bigotimes \bC[\partial^n\phi]$. But the $n$th factor in this product contributes
$$ \tr_{\bC[\partial^n\phi]}(q^{L_0}) = \sum_{m=0}^\infty \tr_{\bC (\partial^n\phi)^m} q^{L_0} = \sum_{m=0}^\infty q^{m(n+1)} = \frac1{1-q^{n+1}},$$
and so
\begin{equation}\label{eqn:bos} Z(\operatorname{Bos}(1)) = q^{-1/24} \prod_{n=0}^\infty \frac1{1-q^{n+1}} = \frac1{\eta}\end{equation}
where $\eta = q^{1/24} \prod_{n=0}^\infty (1-q^{n+1}) = \sqrt[24]{\Delta}$ is Dedekind's eta function.

The final factor depends on the lattice $\cL$ and the vector $\vec\lambda$, and so doesn't have a closed form in general. It is called the \emph{$\vec\lambda$-twisted Theta function} of $\cL$:
$$\Theta^\pm_{\cL + \vec\lambda}(q) := \sum_{\vec\ell \in \cL} (\pm 1)^{\|\vec\ell\|^2} q^{\|\vec\ell+\vec\lambda\|^2/2}.$$
Note that the usual Theta function $\Theta_\cL$ is the special case $\vec\lambda=0$, and that we have included a sign $\pm$ in the notation since we are working with odd lattices: $\Theta^+$ appears in $Z^\NS$ and $\Theta^-$ appears in $Z^\rR$.

In the special case $\cL = D_{12}^+$, we can write a closed formula. This is because $D_{12}^+$ is almost the $\bZ^{12}$-lattice, and so its Theta series almost factor as 12th powers of Theta series for $\bZ$. The twisted Theta series of $\bZ$ are the lower-case twisted theta series:
$$ \theta^\pm_\lambda(q) := \sum_{\ell \in \bZ} (\pm 1)^\ell q^{(\ell + \lambda)^2/2}.$$ 
The original Jacobi theta series arise when $\lambda = 0,\frac12$:
\begin{align*}
  \theta_1 & = \theta^-_{1/2} = 0 & \theta_2 & = \theta^+_{1/2}\\
  \theta_4 & = \theta^-_0 & \theta_3 & = \theta^+_0
\end{align*}
Actually, the usual formula for $\theta_1$ differs from our formula for $\theta^-_{1/2}$ by an overall phase, which in any case doesn't matter because the sum vanishes identically.

Suppose that we cared about $\cL = \bZ^{12}$. Then the twisted Theta function would factor:
$$ \Theta^\pm_{\bZ^{12} + \vec \lambda}(q) = \sum_{\vec\ell \in \bZ^{12}} (\pm 1)^{\|\vec\ell\|^2} q^{\|\vec\ell + \vec\lambda\|^2/2} = \prod_{i=1}^{12} \theta^\pm_{\lambda_i}(q).$$
Observe that the signs of the $\vec\ell$th summand in $\Theta^+$ and $\Theta^-$ agree exactly when $\vec\ell \in D_{12}$ and exactly cancel when $\vec\ell \in \bZ^{12} \smallsetminus D_{12}$. This means that the twisted Theta series for $D_{12}$ is:
\begin{equation}\label{eqn:thetaD12}
 \Theta_{D_{12} + \vec\lambda}(q) = \frac12 \left(\Theta^+_{\bZ^{12}+\vec\lambda}(q) + \Theta^-_{\bZ^{12} + \vec\lambda}(q)\right) = \frac12 \left( \prod_{i=1}^{12} \theta^+_{\lambda_i}(q) + \prod_{i=1}^{12} \theta^-_{\lambda_i}(q)\right).
\end{equation}
Note that $D_{12}$ is even (every $\vec\ell \in D_{12}$ has $\|\vec\ell\|^2 \in 2\bZ$), so there are no signs.

On the other hand, our desired lattice $D_{12}^+$ is nothing but $D_{12}$ together with a shift thereof through $\vec\varsigma := (\frac12,\dots,\frac12)$. Thus
\begin{equation}\label{eqn:thetaD12+} \Theta^\pm_{D_{12}^+ + \vec\lambda} = \Theta_{D_{12} + \vec\lambda} \pm \Theta_{D_{12} + \vec\varsigma + \vec\lambda}.\end{equation}

The very last step is to understand which log-eigenvalues correspond to the canonical parity operator $\rR = (-1)^f \in \Aut(V^{f\natural}) = \SO^+_{24}$. The answer is that $\rR$ can be represented by any vector $\vec\rho \in \bR^{12}$ which has integral dot product with $D_{12}$ and half-integral dot product with $D_{12}^+ \smallsetminus D_{12} = D_{12} + \vec\varsigma$. An example of such a vector is $(0,\dots,0,1)$. Thus if $g$ has log-eigenvalue $\vec\lambda$, then $\rR g$ has log-eigenvalue $\vec\lambda + (0,\dots,0,1)$. Finally, note that such a shift acts very simply on twisted theta functions:
\begin{equation}\label{eqn:shift}
 \theta^\pm_{\lambda+1}(q) = \pm \theta^\pm_\lambda(q)
\end{equation}

Combining equations (\ref{eqn:latticepartition}--\ref{eqn:shift}) then yields our final formula:
\begin{prop}\label{prop:finalformulae}
  Suppose (a lift to $\Co_0$ of)  $g \in \Co_1$ has log-eigenvalues $\vec\lambda \in \bR^{12}/D_{12}$. Then
  \begin{align*}
    Z^\NS_{\rR g}(V^{f\natural}) &= \frac12 \left( \prod_{i=1}^{12} \theta^+_{\lambda_i}(q) - \prod_{i=1}^{12} \theta^-_{\lambda_i}(q) + \prod_{i=1}^{12} \theta^+_{\lambda_i+\frac12}(q) - \prod_{i=1}^{12} \theta^-_{\lambda_i+\frac12}(q)\right) \\
    Z^\rR_{\rR g}(V^{f\natural}) &= \frac12 \left( \prod_{i=1}^{12} \theta^+_{\lambda_i}(q) - \prod_{i=1}^{12} \theta^-_{\lambda_i}(q) - \prod_{i=1}^{12} \theta^+_{\lambda_i+\frac12}(q) + \prod_{i=1}^{12} \theta^-_{\lambda_i+\frac12}(q)\right)
  \end{align*}\qed
\end{prop}

\section{Results} \label{sec:results}

The following table reports the results of our computations, which were done using  Maple. We ran through all conjugacy classes $g\in \Co_0$. (The answer only depends on the corresponding class in $\Co_1$, but we ran through all classes in $\Co_0$ in order to double-check our results.) For each class, we used the Frame shape computed by \cite{kondo} to produce log-eigenvalues to parameterize the class of $g$ in $\rO_{24}$, and hence the log eigenvalues of the four possible lifts of $g$ to $\Spin_{24}$. For each of these lifts, we computed $Z^\rR_{\rR g}(V^{f\natural})$ from Lemma~\ref{prop:finalformulae}, and found that only one produced $\tr_{\bR^{24}}(g)$. For that lift, we computed $Z^\NS_{\rR g}(V^{f\natural})$ from Lemma~\ref{prop:finalformulae}. In the table, we have reported the Frame shapes and the computed values of $Z^\rR_{\rR g}(V^{f\natural})$ and $Z^\NS_{\rR g}(V^{f\natural})$. We also record the names of $g \in \Co_0$ and of its image in $\Co_1$ following the ATLAS naming convention: a conjugacy class is named $nX$ where $n \in \bN$ records the order of $g$, and $X \in \{\text A,\text B,\text C,\dots\}$ sorts the conjugacy classes of a given order in terms of the size of their centralizers, starting from the largest. 

The penultimate column of the table records whether $g$ is conspiratorial (\checkmark), nonconspiratorial (\textbf{X}), or spontaneously breaks supersymmetry~(0) --- as explained in Section~\ref{sec:preliminaries}, these are read off by comparing $Z^{\rR}_{\rR g}$ to the constant term of $Z^\NS_{\rR g}$. 
As reported in the introduction, we find that the only conspiratorial $\Co_1$-classes are
$$ 3C, \quad 4B, \quad 5C, \quad 6F, \quad 8D, \quad 9B.$$

In the final column of the table, included purely for the reader's convenience, we record from \cite{Theo} whether the class is  (\checkmark) or is not~(~) anomalous.
Although it was not to study the relationship between anomalies and spontaneous symmetry breaking, we find that a class in $\Co_1$ is anomalous if and only if it spontaneously breaks supersymmetry. Actually, this is immediate from our observation in the introduction that the conspiratorial classes are all unbalanced. Indeed, if $g$ is an anomalous symmetry of any holomorphic SCFT $V$, then $Z^\rR_{\rR g}(V) \propto Z^{\rR g}_\rR(V)$ must be a constant but must transform with nontrivial eigenvalue under $T$-transformations, and so must vanish; specializing to $V^{f\natural}$, we see that a necessary condition for an anomaly is that $g$ must be traceless. On the other hand, we found by calculation that the only conspiratorial classes are unbalanced, and~\cite{Theo} found by calculation that all unbalanced classes are non-anomalous. So if $g$ is balanced and traceless, then $Z^{\NS}_{\rR g} = |Z^{\rR}_{\rR g}| + O(q) = 0 + O(q)$, and so $V_{\rR g}$ has spontaneous supersymmetry breaking.

\begin{landscape}
\tiny

\begin{longtable}{|c|c|c|c|l|c|c|}
\hline
&&&&&& \\[-6pt]
$\Co_0$ & $\Co_1$ & Frame shape & $Z^\rR_{\rR g}$ & $Z^\NS_{\rR g}$ & Consp? &Anom? \\[2pt]
\hline
\endhead

\hline\endfoot

1A & 1A &$ 1^{24} $&  24 & $24+4096 q +98304 q^{2}+1228800 q^{3}+10747904 q^{4}+74244096 q^{5}+432144096 q^{6}+ \dots$ & \checkmark&
\\
 2A & 1A & $2^{24}1^{-24} $& -24 & $24+4096 q +98304 q^{2}+1228800 q^{3}+10747904 q^{4}+74244096 q^{5}+432144096 q^{6}+ \dots$ & \checkmark&
\\
 2B & 2A &$ 1^{8}2^{8} $& 8 & $8+256q^{1/2}+2048q+11264q^{3/2}+49152q^2+183808q^{5/2}+614400q^3+1882112q^{7/2}+ \dots$ & \checkmark&
\\
 2C & 2A& $2^{16}1^{-8}$ & -8 & $8+256q^{1/2}+2048q+11264q^{3/2}+49152q^2+183808q^{5/2}+614400q^3+1882112q^{7/2}+ \dots$ & \checkmark&
\\
 4A & 2B &$4^{12}2^{-12}$ & 0 & $0+24q^{1/8}+464q^{5/8}+3192q^{9/8}+16656q^{13/8}+69040q^{17/8}+251136q^{21/8}+ \dots$&0&\checkmark
\\
 2D & 2C &$2^{12} $& 0 & $0+64q^{1/4}+768q^{3/4}+4992q^{5/4}+24064q^{7/4}+96576q^{9/4}+340224q^{11/4}+ \dots$&0&\checkmark
\\
 3A & 3A &$3^{12}1^{-12} $& -12 & $12+24q^{1/3}+440q^{2/3}+1608q+3192q^{4/3}+13904q^{5/3}+35688q^2+68160q^{7/3}+ \dots$& \checkmark&
\\
 6A & 3A &$1^{12}6^{12}2^{-12}3^{-12}$ & 12 & $12+24q^{1/3}+440q^{2/3}+1608q+3192q^{4/3}+13904q^{5/3}+35688q^2+68160q^{7/3}+ \dots$& \checkmark&
\\
 3B & 3B &$1^{6}3^{6}$ & 6 & $6+64q^{1/3}+384q^{2/3}+1344q+4352q^{4/3}+12672q^{5/3}+32640q^2+79872q^{7/3}+\dots$& \checkmark&
\\
 6B & 3B &$2^{6}6^{6}1^{-6}3^{-6}$ & -6 & $6+64q^{1/3}+384q^{2/3}+1344q+4352q^{4/3}+12672q^{5/3}+32640q^2+79872q^{7/3}+ \dots$& \checkmark&
\\
 {3C} & {3C}& $3^{9}1^{-3}$  & {-3} & ${5}+72q^{1/3}+360q^{2/3}+1368q+4392q^{4/3}+12528q^{5/3}+32760q^2+80064q^{7/3}+ \dots$ &\textbf{X}&
\\
  {6C} &  {3C} & $1^{3}6^{9}2^{-3}3^{-9}$ & {3} &  ${5}+72q^{1/3}+360q^{2/3}+1368q+4392q^{4/3}+12528q^{5/3}+32760q^2+80064q^{7/3}+ \dots$&\textbf{X}&
\\
 3D & 3D &$3^{8} $& 0 & $0+16q^{1/9}+128q^{4/9}+576q^{7/9}+2048q^{10/9}+6304q^{13/9}+17408q^{16/9}+44416q^{19/9}+ \dots$&0&\checkmark
\\
 6D & 3D &$6^{8}3^{-8}$ & 0 & $0+16q^{1/9}+128q^{4/9}+576q^{7/9}+2048q^{10/9}+6304q^{13/9}+17408q^{16/9}+44416q^{19/9}+ \dots$&0&\checkmark
\\
 4B & 4A &$1^{8}4^{8}2^{-8}$ & 8 & $8+16q^{1/4}+128q^{1/2}+448q^{3/4}+1024q+2272q^{5/4}+5632q^{3/2}+12672q^{7/4}+24576q^2+\dots$& \checkmark&
\\
 4C & 4A &$ 4^{8}1^{-8} $& -8 & $8+16q^{1/4}+128q^{1/2}+448q^{3/4}+1024q+2272q^{5/4}+5632q^{3/2}+12672q^{7/4}+24576q^2+ \dots$& \checkmark&
\\
  {4D} &  {4B} & $ 4^{8}2^{-4}$ & {0} &  ${4}+32q^{1/4}+128q^{1/2}+384q^{3/4}+1024q+2496q^{5/4}+5632q^{3/2}+12032q^{7/4}+24576q^2+ \dots$&\textbf{X}&
\\
 4E & 4C & $1^{4}2^{2}4^{4}$ & 4 & $4+32q^{1/4}+128q^{1/2}+384q^{3/4}+1024q+2496q^{5/4}+5632q^{3/2}+12032q^{7/4}+24576q^2+ \dots$& \checkmark&
\\
 4F & 4C & $2^{6}4^{4}1^{-4} $& -4 & $4+32q^{1/4}+128q^{1/2}+384q^{3/4}+1024q+2496q^{5/4}+5632q^{3/2}+12032q^{7/4}+24576q^2+ \dots$& \checkmark&
\\

 4G & 4D & $2^{4}4^{4}$ & 0 & $0+16q^{1/8}+64q^{3/8}+224q^{5/8}+640q^{7/8}+1616q^{9/8}+3776q^{11/8}+8288q^{13/8}+ \dots$&0&\checkmark
\\
 8A & 4E  & $8^{6}4^{-6}$ & 0 & $0+12q^{3/32}+52q^{11/32}+204q^{19/32}+552q^{27/32}+1456q^{35/32}+3396q^{43/32}+7560q^{51/32}+ \dots$&0&\checkmark
\\

 4H & 4F &$ 4^{6}$ & 0 & $0+8q^{1/16}+48q^{5/16}+168q^{9/16}+496q^{13/16}+1296q^{17/16}+3072q^{21/16}+6840q^{25/16}+\dots$&0&\checkmark
\\
 5A & 5A & $5^{6}1^{-6}$ & -6 & $6+12q^{1/5}+52q^{2/5}+192q^{3/5}+372q^{4/5}+844q+1584q^{6/5}+3216q^{7/5}+6388q^{8/5}+ \dots$& \checkmark&
\\
 10A & 5A & $1^{6}10^{6}2^{-6}5^{-6}$ & 6 & $6+12q^{1/5}+52q^{2/5}+192q^{3/5}+372q^{4/5}+844q+1584q^{6/5}+3216q^{7/5}+6388q^{8/5}+\dots$& \checkmark&
\\
 5B & 5B & $1^{4}5^{4}$ & 4 & $4+16q^{1/5}+64q^{2/5}+160q^{3/5}+384q^{4/5}+816q+1664q^{6/5}+3296q^{7/5}+6144q^{8/5}+1 \dots$& \checkmark&
\\
 10B & 5B & $2^{4}10^{4}1^{-4}5^{-4}$ & -4 & $4+16q^{1/5}+64q^{2/5}+160q^{3/5}+384q^{4/5}+816q+1664q^{6/5}+3296q^{7/5}+6144q^{8/5}+ \dots$& \checkmark&
\\
  {5C} &  {5C} &  $5^{5}1^{-1}$ & {-1} &  ${3}+20q^{1/5}+60q^{2/5}+160q^{3/5}+380q^{4/5}+820q+1680q^{6/5}+3280q^{7/5}+6140q^{8/5}+ \dots$&\textbf{X}&
\\
  {10C} &  {5C} & $1^{1}10^{5}2^{-1}5^{-5} $&  {1} &  ${3}+20q^{1/5}+60q^{2/5}+160q^{3/5}+380q^{4/5}+820q+1680q^{6/5}+3280q^{7/5}+6140q^{8/5}+ \dots$&\textbf{X}&
\\
 6E & 6A & $3^{4}6^{4}1^{-4}2^{-4} $& -4 & $4+16q^{1/6}+8q^{1/3}+112q^{1/2}+232q^{2/3}+224q^{5/6}+792q+1408q^{7/6}+1576q^{4/3}+ \dots$& \checkmark&
\\

 6F & 6A &$ 1^{4}6^{8}2^{-8}3^{-4}$ & 4 & $4+16q^{1/6}+8q^{1/3}+112q^{1/2}+232q^{2/3}+224q^{5/6}+792q+1408q^{7/6}+1576q^{4/3}+ \dots$& \checkmark&
\\
 12A & 6B &$ 2^{6}12^{6}4^{-6}6^{-6}$ & 0 & $0+12q^{1/8}+40q^{7/24}+24q^{11/24}+192q^{5/8}+360q^{19/24}+384q^{23/24}+1236q^{9/8}+ \dots$&0&\checkmark
\\
 6G & 6C &$ 2^{5}3^{4}6^{1}1^{-4}$ & -4 & $4+8q^{1/6}+40q^{1/3}+88q^{1/2}+168q^{2/3}+368q^{5/6}+696q+1216q^{7/6}+2216q^{4/3}+ \dots$& \checkmark&
\\
 6H & 6C & $1^{4}2^{1}6^{5}3^{-4}$ & 4 & $4+8q^{1/6}+40q^{1/3}+88q^{1/2}+168q^{2/3}+368q^{5/6}+696q+1216q^{7/6}+2216q^{4/3}+\dots$& \checkmark&
\\
 6I & 6D & $1^{5}3^{1}6^{4}2^{-4}$ & 5 & $5+8q^{1/6}+32q^{1/3}+88q^{1/2}+192q^{2/3}+368q^{5/6}+672q+1216q^{7/6}+2176q^{4/3}+ \dots$& \checkmark&
\\
6J & 6D & $2^{1}6^{5}1^{-5}3^{-1}$ & -5 & $5+8q^{1/6}+32q^{1/3}+88q^{1/2}+192q^{2/3}+368q^{5/6}+672q+1216q^{7/6}+2176q^{4/3}+ \dots$& \checkmark&
\\
 6K & 6E & $1^{2}2^{2}3^{2}6^{2}$ & 2 & $2+16q^{1/6}+32q^{1/3}+80q^{1/2}+192q^{2/3}+352q^{5/6}+672q+1280q^{7/6}+2176q^{4/3}+ \dots$& \checkmark&
\\
 6L & 6E & $2^{4}6^{4}1^{-2}3^{-2}$ & -2 & $2+16q^{1/6}+32q^{1/3}+80q^{1/2}+192q^{2/3}+352q^{5/6}+672q+1280q^{7/6}+2176q^{4/3}+\dots$& \checkmark&
\\
  {6M} &  {6F} &  $3^{3}6^{3}1^{-1}2^{-1}$ & {-1} & ${3}+12q^{1/6}+36q^{1/3}+84q^{1/2}+180q^{2/3}+360q^{5/6}+684q+1248q^{7/6}+2196q^{4/3}+ \dots$&\textbf{X}&
\\
  {6N} &  {6F} &  $1^{1}6^{6}2^{-2}3^{-3}$ & {1} &  ${3}+12q^{1/6}+36q^{1/3}+84q^{1/2}+180q^{2/3}+360q^{5/6}+684q+1248q^{7/6}+2196q^{4/3}+  \dots$&\textbf{X}&
\\
 6O & 6G & $2^{3}6^{3} $& 0 & $0+8q^{1/12}+24q^{1/4}+48q^{5/12}+128q^{7/12}+264q^{3/4}+480q^{11/12}+944q^{13/12}+1680q^{5/4}+ \dots$&0&\checkmark
\\

 12B & 6H & $12^{4}6^{-4} $& 0 & $0+8q^{5/72}+16q^{17/72}+56q^{29/72}+112q^{41/72}+248q^{53/72}+464q^{65/72}+896q^{77/72}+ \dots$&0&\checkmark
\\
 6P & 6I &$ 6^{4}$ & 0 & $0+4q^{1/36}+16q^{7/36}+40q^{13/36}+96q^{19/36}+204q^{25/36}+400q^{31/36}+760q^{37/36}+ \dots$&0&\checkmark
\\
 7A & 7A & $7^{4}1^{-4}$ & -4 & $4+8q^{1/7}+16q^{2/7}+56q^{3/7}+104q^{4/7}+200q^{5/7}+328q^{6/7}+592q+976q^{8/7}+1560q^{9/7}+ \dots$& \checkmark&
\\
 14A & 7A & $1^{4}14^{4}2^{-4}7^{-4}$ & 4 & $4+8q^{1/7}+16q^{2/7}+56q^{3/7}+104q^{4/7}+200q^{5/7}+328q^{6/7}+592q+976q^{8/7}+1560q^{9/7}+ \dots$& \checkmark&
\\
 7B & 7B & $1^{3}7^{3}$ & 3 & $3+8q^{1/7}+24q^{2/7}+48q^{3/7}+104q^{4/7}+192q^{5/7}+336q^{6/7}+584q+984q^{8/7}+1608q^{9/7}+\dots$& \checkmark&
\\
 14B & 7B & $ 2^{3}14^{3}1^{-3}7^{-3}$ & -3 & $3+8q^{1/7}+24q^{2/7}+48q^{3/7}+104q^{4/7}+192q^{5/7}+336q^{6/7}+584q+984q^{8/7}+1608q^{9/7}+\dots$& \checkmark&
\\
 8B & 8A & $8^{4}2^{-4}$ &0 & $0+8q^{1/16}+8q^{3/16}+16q^{5/16}+48q^{7/16}+104q^{9/16}+152q^{11/16}+208q^{13/16}+400q^{15/16}+ \dots$&0&\checkmark
\\
 8C & 8B &  $2^{4}8^{4}4^{-4}$ &0 & $0+4q^{1/16}+16q^{3/16}+24q^{5/16}+32q^{7/16}+84q^{9/16}+176q^{11/16}+248q^{13/16}+352q^{15/16}+ \dots$&0&\checkmark
\\
 8D & 8C & $1^{4}8^{4}2^{-2}4^{-2}$ &  4 & $4+4q^{1/8}+16q^{1/4}+32q^{3/8}+64q^{1/2}+120q^{5/8}+192q^{3/4}+320q^{7/8}+512q+788q^{9/8}+ \dots$& \checkmark&
\\
 8E & 8C &$ 2^{2}8^{4}1^{-4}4^{-2}$ & -4 & $4+4q^{1/8}+16q^{1/4}+32q^{3/8}+64q^{1/2}+120q^{5/8}+192q^{3/4}+320q^{7/8}+512q+788q^{9/8}+ \dots$& \checkmark&
\\
  {8F}&  {8D} &$ 8^{4}4^{-2}$ &   {0} &  ${2}+8q^{1/8}+16q^{1/4}+32q^{3/8}+64q^{1/2}+112q^{5/8}+192q^{3/4}+320q^{7/8}+512q+808q^{9/8}+ \dots$&\textbf{X}&
\\
 8G & 8E & $1^{2}2^{1}4^{1}8^{2} $&  2 & $2+8q^{1/8}+16q^{1/4}+32q^{3/8}+64q^{1/2}+112q^{5/8}+192q^{3/4}+320q^{7/8}+512q+808q^{9/8}+ \dots$& \checkmark&
\\
 8H & 8E & $2^{3}4^{1}8^{2}1^{-2}$ &  -2 & $2+8q^{1/8}+16q^{1/4}+32q^{3/8}+64q^{1/2}+112q^{5/8}+192q^{3/4}+320q^{7/8}+512q+808q^{9/8}+ \dots$& \checkmark&
\\
 8I & 8F & $4^{2}8^{2}$ &  0 & $0+4q^{1/32}+8q^{5/32}+20q^{9/32}+40q^{13/32}+72q^{17/32}+128q^{21/32}+220q^{25/32}+360q^{29/32}+ \dots$&0&\checkmark
\\
 9A & 9A &$ 9^{3}1^{-3}$ &  -3 & $3+6q^{1/9}+8q^{2/9}+24q^{1/3}+42q^{4/9}+80q^{5/9}+120q^{2/3}+192q^{7/9}+296q^{8/9}+456q+ \dots$& \checkmark&
\\
 18A & 9A & $1^{3}18^{3}2^{-3}9^{-3}$ &  3 & $3+6q^{1/9}+8q^{2/9}+24q^{1/3}+42q^{4/9}+80q^{5/9}+120q^{2/3}+192q^{7/9}+296q^{8/9}+456q+ \dots$ & \checkmark&
\\
  {9B} &  {9B} &  $9^{3}3^{-1} $&  {0} &  ${2}+6q^{1/9}+12q^{2/9}+24q^{1/3}+42q^{4/9}+72q^{5/9}+120q^{2/3}+192q^{7/9}+300q^{8/9}+456q+ \dots$&\textbf{X}&
\\
  {18B} &  {9B} &  $3^{1}18^{3}6^{-1}9^{-3}$ &  {0} &  ${2}+6q^{1/9}+12q^{2/9}+24q^{1/3}+42q^{4/9}+72q^{5/9}+120q^{2/3}+192q^{7/9}+300q^{8/9}+456q+ \dots$ &\textbf{X}&
\\
 9C & 9C &$ 1^{3}9^{3}3^{-2} $&  3 & $3+4q^{1/9}+12q^{2/9}+24q^{1/3}+44q^{4/9}+72q^{5/9}+120q^{2/3}+192q^{7/9}+300q^{8/9}+456q+ \dots$ & \checkmark&
\\

18C &9C & $2^{3}3^{2}18^{3}1^{-3}6^{-2}9^{-3}$ &  -3 & $3+4q^{1/9}+12q^{2/9}+24q^{1/3}+44q^{4/9}+72q^{5/9}+120q^{2/3}+192q^{7/9}+300q^{8/9}+456q+ \dots$& \checkmark&
\\

10D &10A &$5^{2}10^{2}1^{-2}2^{-2}$  & -2 & $2+8q^{1/10}+4q^{1/5}+16q^{3/10}+28q^{2/5}+56q^{1/2}+96q^{3/5}+112q^{7/10}+188q^{4/5}+ \dots$ & \checkmark&
\\
 10E &10A &$1^{2}10^{4}2^{-4}5^{-2}$ &  2 & $2+8q^{1/10}+4q^{1/5}+16q^{3/10}+28q^{2/5}+56q^{1/2}+96q^{3/5}+112q^{7/10}+188q^{4/5}+\dots$ & \checkmark&
\\
20A &10B &$2^{3}20^{3}4^{-3}10^{-3}$ &  0 &$0+6q^{3/40}+8q^{7/40}+18q^{11/40}+24q^{3/8}+32q^{19/40}+84q^{23/40}+114q^{27/40}+180q^{31/40}+ \dots$&0&\checkmark
\\
20B &10C &$4^{2}20^{2}2^{-2}10^{-2}$ &  0 &$0+4q^{1/40}+4q^{1/8}+12q^{9/40}+16q^{13/40}+40q^{17/40}+56q^{21/40}+92q^{5/8}+136q^{29/40}+\dots$&0&\checkmark
\\
10F &10D &$2^{3}5^{2}10^{1}1^{-2}$ &  -2 &$2+4q^{1/10}+12q^{1/5}+16q^{3/10}+28q^{2/5}+52q^{1/2}+80q^{3/5}+128q^{7/10}+188q^{4/5}+ \dots$ & \checkmark&
\\
10G &10D &$1^{2}2^{1}10^{3}5^{-2}$ &  2 &$2+4q^{1/10}+12q^{1/5}+16q^{3/10}+28q^{2/5}+52q^{1/2}+80q^{3/5}+128q^{7/10}+188q^{4/5}+\dots$& \checkmark&
\\
10H & 10E &$1^{3}5^{1}10^{2}2^{-2}$ &  3 &$3+4q^{1/10}+8q^{1/5}+16q^{3/10}+32q^{2/5}+52q^{1/2}+80q^{3/5}+128q^{7/10}+192q^{4/5}+ \dots$ & \checkmark&
\\
10I & 10E &$2^{1}10^{3}1^{-3}5^{-1}$ &  -3 &$3+4q^{1/10}+8q^{1/5}+16q^{3/10}+32q^{2/5}+52q^{1/2}+80q^{3/5}+128q^{7/10}+192q^{4/5}+ \dots$ & \checkmark&
\\
10J &10F &$2^{2}10^{2}$ &  0 &$0+4q^{1/20}+8q^{3/20}+12q^{1/4}+24q^{7/20}+36q^{9/20}+64q^{11/20}+104q^{13/20}+152q^{3/4}+ \dots$&0&\checkmark
\\
11A &11A & $1^{2}11^{2}$ &  2 &$2+4q^{1/11}+8q^{2/11}+12q^{3/11}+24q^{4/11}+36q^{5/11}+56q^{6/11}+88q^{7/11}+128q^{8/11}+ \dots$& \checkmark&
\\
22A &11A & $2^{2}22^{2}1^{-2}11^{-2}$ &  0 &$0+2q^{1/44}+2q^{3/44}+2q^{5/44}+4q^{7/44}+4q^{9/44}+6q^{1/4}+8q^{13/44}+10q^{15/44}+12q^{17/44}+ \dots$&0&\checkmark
\\
12C &12A &$2^{4}3^{4}12^{4}1^{-4}4^{-4}6^{-4}$ &  -4 &$4+8q^{1/6}+8q^{1/4}+8q^{1/3}+48q^{5/12}+56q^{1/2}+16q^{7/12}+104q^{2/3}+176q^{3/4}+112q^{5/6}+ \dots$& \checkmark&
\\
12D &12A & $1^{4}12^{4}3^{-4}4^{-4}$ &  4 &$4+8q^{1/6}+8q^{1/4}+8q^{1/3}+48q^{5/12}+56q^{1/2}+16q^{7/12}+104q^{2/3}+176q^{3/4}+112q^{5/6}+ \dots$& \checkmark&
\\
 12E &12B & $2^{2}12^{4}4^{-4}6^{-2}$ &  0 &$0+8q^{1/12}+4q^{1/6}+8q^{1/4}+28q^{1/3}+16q^{5/12}+28q^{1/2}+96q^{7/12}+76q^{2/3}+\dots$&0&\checkmark
\\
12F &12C & $6^{2}12^{2}2^{-2}4^{2}$ &  0 &$0+4q^{1/24}+4q^{1/8}+8q^{5/24}+16q^{7/24}+20q^{3/8}+32q^{11/24}+56q^{13/24}+72q^{5/8}+ \dots$&0&\checkmark
\\
12G &12D &$2^{1}3^{3}12^{3}1^{-1}4^{-1}6^{-3}$ &  -1 &$1+6q^{1/12}+6q^{1/6}+6q^{1/4}+18q^{1/3}+36q^{5/12}+42q^{1/2}+48q^{7/12}+90q^{2/3}+150q^{3/4}+\dots$& \checkmark&
\\
12H &12D &$ 1^{1}12^{3}3^{-3}4^{-1} $&  1 &$1+6q^{1/12}+6q^{1/6}+6q^{1/4}+18q^{1/3}+36q^{5/12}+42q^{1/2}+48q^{7/12}+90q^{2/3}+150q^{3/4}+\dots$& \checkmark&
\\
12I & 12E & $1^{2}3^{2}4^{2}12^{2}2^{-2}6^{-2}$ &  2 &$2+4q^{1/12}+8q^{1/6}+4q^{1/4}+16q^{1/3}+40q^{5/12}+40q^{1/2}+48q^{7/12}+96q^{2/3}+148q^{3/4}+\dots$& \checkmark&
\\
12J & 12E & $4^{2}12^{2}1^{-2}3^{-2}$ & -2 &$2+4q^{1/12}+8q^{1/6}+4q^{1/4}+16q^{1/3}+40q^{5/12}+40q^{1/2}+48q^{7/12}+96q^{2/3}+148q^{3/4}+\dots$& \checkmark&
\\
24A &12F & $4^{3}24^{3}8^{-3}12^{-3}$ & 0 &$0+6q^{5/96}+2q^{13/96}+6q^{7/32}+24q^{29/96}+18q^{37/96}+24q^{15/32}+84q^{53/96}+56q^{61/96}+ \dots$&0&\checkmark
\\
12K &12G & $4^{2}12^{2}2^{-1}6^{-1}$ &  0 &$0+2q^{1/144}+4q^{13/144}+6q^{25/144}+12q^{37/144}+18q^{49/144}+28q^{61/144}+44q^{73/144}+\dots$&0&\checkmark
\\

12I &12H &$ 1^{1}2^{2}3^{1}12^{2}4^{-2}$ &  1 &$1+4q^{1/12}+8q^{1/6}+12q^{1/4}+16q^{1/3}+24q^{5/12}+40q^{1/2}+64q^{7/12}+96q^{2/3}+132q^{3/4}+\dots$& \checkmark&
\\
12M &12H &$2^{3}6^{1}12^{2}1^{-1}3^{-1}4^{-2}$ &  -1 &$1+4q^{1/12}+8q^{1/6}+12q^{1/4}+16q^{1/3}+24q^{5/12}+40q^{1/2}+64q^{7/12}+96q^{2/3}+132q^{3/4}+ \dots$& \checkmark&
\\
12N &12I & $2^{2}3^{2}4^{1}12^{1}1^{-2}$ &  -2 &$2+4q^{1/12}+4q^{1/6}+12q^{1/4}+20q^{1/3}+24q^{5/12}+44q^{1/2}+64q^{7/12}+84q^{2/3}+132q^{3/4}+\dots$& \checkmark&
\\
12O &12I & $1^{2}4^{1}6^{2}12^{1}3^{-2}$ &  2 &$2+4q^{1/12}+4q^{1/6}+12q^{1/4}+20q^{1/3}+24q^{5/12}+44q^{1/2}+64q^{7/12}+84q^{2/3}+132q^{3/4}+ \dots$& \checkmark&
\\
12P &12J & $2^{1}4^{1}6^{1}12^{1}$ &  0 &$0+4q^{1/24}+4q^{1/8}+8q^{5/24}+16q^{7/24}+20q^{3/8}+32q^{11/24}+56q^{13/24}+72q^{5/8}+ \dots$&0&\checkmark
\\
12Q &12K &$1^{3}12^{3}2^{-1}3^{-1}4^{-1}6^{-1}$ &  3 &$3+2q^{1/12}+6q^{1/6}+10q^{1/4}+18q^{1/3}+28q^{5/12}+42q^{1/2}+64q^{7/12}+90q^{2/3}+130q^{3/4}+ \dots$& \checkmark&
\\
12R &12K &$2^{2}3^{1}12^{3}1^{-3}4^{-1}6^{-2}$ &  -3 &$3+2q^{1/12}+6q^{1/6}+10q^{1/4}+18q^{1/3}+28q^{5/12}+42q^{1/2}+64q^{7/12}+90q^{2/3}+130q^{3/4}+ \dots$& \checkmark&
\\
24B &12L & $24^{2}12^{-2}$ &  0 &$0+4q^{11/288}+4q^{35/288}+8q^{59/288}+12q^{83/288}+24q^{107/288}+32q^{131/288}+52q^{155/288}+\dots$&0&\checkmark
\\

12S &12M & $12^{2}$ &  0 &$0+2q^{1/144}+4q^{13/144}+6q^{25/144}+12q^{37/144}+18q^{49/144}+28q^{61/144}+44q^{73/144}+\dots$&0&\checkmark
\\
13A &13A & $13^{2}1^{-2}$ &  -2 &$2+4q^{1/13}+4q^{2/13}+8q^{3/13}+12q^{4/13}+24q^{5/13}+32q^{6/13}+48q^{7/13}+68q^{8/13}+ \dots$& \checkmark&
\\
26A &13A & $1^{2}26^{2}2^{-2}13^{-2}$ &  2 &$2+4q^{1/13}+4q^{2/13}+8q^{3/13}+12q^{4/13}+24q^{5/13}+32q^{6/13}+48q^{7/13}+68q^{8/13}+ \dots$ & \checkmark&
\\
28A &14A & $2^{2}28^{2}4^{-2}14^{-2}$ &  0 &$0+4q^{3/56}+4q^{1/8}+8q^{11/56}+8q^{15/56}+20q^{19/56}+16q^{23/56}+32q^{27/56}+48q^{31/56}+ \dots$&0&\checkmark
\\
14C &14B & $1^{1}2^{1}7^{1}14^{1}$ &  1 &$1+4q^{1/14}+4q^{1/7}+8q^{3/14}+12q^{2/7}+16q^{5/14}+24q^{3/7}+36q^{1/2}+52q^{4/7}+68q^{9/14}+ \dots$& \checkmark&
\\
14D &14B & $2^{2}14^{2}1^{-1}7^{-1}$ &  -1 &$1+4q^{1/14}+4q^{1/7}+8q^{3/14}+12q^{2/7}+16q^{5/14}+24q^{3/7}+36q^{1/2}+52q^{4/7}+68q^{9/14}+ \dots$& \checkmark&
\\
15A &15A &$1^{3}15^{3}3^{-3}5^{-3}$ &  3 &$3+2q^{1/15}+6q^{1/5}+18q^{4/15}+6q^{1/3}+24q^{2/5}+36q^{7/15}+12q^{8/15}+78q^{3/5}+92q^{2/3}+\dots$ & \checkmark&
\\
30A &15A & $2^{3}3^{3}5^{3}30^{3}1^{-3}6^{-3}10^{-3}15^{-3}$ & -3 &$3+2q^{1/15}+6q^{1/5}+18q^{4/15}+6q^{1/3}+24q^{2/5}+36q^{7/15}+12q^{8/15}+78q^{3/5}+92q^{2/3}+ \dots$& \checkmark&
\\
15B &15B & $3^{2}15^{2}1^{-2}5^{-2}$ &  -2 &$2+4q^{1/15}+8q^{1/5}+12q^{4/15}+4q^{1/3}+28q^{2/5}+40q^{7/15}+16q^{8/15}+68q^{3/5}+88q^{2/3}+\dots$ & \checkmark&
\\
30B &15B &$1^{2}5^{2}6^{2}30^{2}2^{-2}3^{-2}10^{-2}15^{-2}$ &  2 &$2+4q^{1/15}+8q^{1/5}+12q^{4/15}+4q^{1/3}+28q^{2/5}+40q^{7/15}+16q^{8/15}+68q^{3/5}+88q^{2/3}+\dots$& \checkmark&
\\
15C &15C &$15^{2}3^{-2}$ &  0 &$0+4q^{2/45}+4q^{1/9}+4q^{8/45}+8q^{11/45}+8q^{14/45}+24q^{17/45}+28q^{4/9}+32q^{23/45}+ \dots$&0&\checkmark
\\
30C &15C & $3^{2}30^{2}6^{-2}15^{-2}$ & 0 &$0+4q^{2/45}+4q^{1/9}+4q^{8/45}+8q^{11/45}+8q^{14/45}+24q^{17/45}+28q^{4/9}+32q^{23/45}+ \dots$&0&\checkmark
\\
15D &15D & $1^{1}3^{1}5^{1}15^{1}$ &  1 &$1+4q^{1/15}+4q^{2/15}+4q^{1/5}+12q^{4/15}+12q^{1/3}+20q^{2/5}+32q^{7/15}+36q^{8/15}+52q^{3/5}+ \dots$& \checkmark&
\\

30D &15D &$2^{1}6^{1}10^{1}30^{1}1^{-1}3^{-1}5^{-1}15^{-1}$ &  -1 &$1+4q^{1/15}+4q^{2/15}+4q^{1/5}+12q^{4/15}+12q^{1/3}+20q^{2/5}+32q^{7/15}+36q^{8/15}+52q^{3/5}+ \dots$& \checkmark&
\\
 15E & 15E & $1^{2}15^{2}3^{-1}5^{-1}$ &  2 &$2+2q^{1/15}+4q^{2/15}+6q^{1/5}+10q^{4/15}+14q^{1/3}+20q^{2/5}+28q^{7/15}+40q^{8/15}+54q^{3/5}+ \dots$& \checkmark&
\\
 30E & 15E &$2^{2}3^{1}5^{1}30^{2}1^{-2}6^{-1}10^{-1}15^{-2}$ &  -2 &$2+2q^{1/15}+4q^{2/15}+6q^{1/5}+10q^{4/15}+14q^{1/3}+20q^{2/5}+28q^{7/15}+40q^{8/15}+54q^{3/5}+ \dots$& \checkmark&
\\
16A &16A & $2^{2}16^{2}4^{-1}8^{-1}$ &  0 &$0+2q^{1/32}+4q^{3/32}+4q^{5/32}+8q^{7/32}+10q^{9/32}+12q^{11/32}+20q^{13/32}+24q^{15/32}+ \dots$&0&\checkmark
\\
16B &16B & $1^{2}16^{2}2^{-1}8^{-1}$ &  2 &$2+2q^{1/16}+4q^{1/8}+4q^{3/16}+8q^{1/4}+12q^{5/16}+16q^{3/8}+24q^{7/16}+32q^{1/2}+42q^{9/16}+ \dots$ & \checkmark&
\\
16C &16B & $2^{1}16^{2}1^{-2}8^{-1}$ &  -2 &$2+2q^{1/16}+4q^{1/8}+4q^{3/16}+8q^{1/4}+12q^{5/16}+16q^{3/8}+24q^{7/16}+32q^{1/2}+42q^{9/16}+ \dots$& \checkmark&
\\
18D &18A & $9^{1}18^{1}1^{-1}2^{-1}$ &  -1 &$1+4q^{1/18}+2q^{1/9}+4q^{1/6}+4q^{2/9}+8q^{5/18}+12q^{1/3}+16q^{7/18}+22q^{4/9}+28q^{1/2}+ \dots$& \checkmark&
\\
 18E &18A & $1^{1}18^{2}2^{-2}9^{-1}$ &  1 &$1+4q^{1/18}+2q^{1/9}+4q^{1/6}+4q^{2/9}+8q^{5/18}+12q^{1/3}+16q^{7/18}+22q^{4/9}+28q^{1/2}+ \dots$& \checkmark&
\\
18F &18B & $1^{2}9^{1}18^{1}2^{-1}3^{-1}$ &  2 &$2+2q^{1/18}+2q^{1/9}+4q^{1/6}+6q^{2/9}+8q^{5/18}+12q^{1/3}+16q^{7/18}+22q^{4/9}+28q^{1/2}+ \dots$& \checkmark&\\
18G &18B &$2^{1}3^{1}18^{2}1^{-2}6^{-1}9^{-1}$ &  -2 &$2+2q^{1/18}+2q^{1/9}+4q^{1/6}+6q^{2/9}+8q^{5/18}+12q^{1/3}+16q^{7/18}+22q^{4/9}+28q^{1/2}+ \dots$& \checkmark&
\\
18H &18C & $2^{2}9^{1}18^{1}1^{-1}6^{-1}$ &  -1 &$1+2q^{1/18}+4q^{1/9}+4q^{1/6}+6q^{2/9}+8q^{5/18}+12q^{1/3}+16q^{7/18}+20q^{4/9}+28q^{1/2}+ \dots$& \checkmark&
\\
18I &18C & $1^{1}2^{1}18^{2}6^{-1}9^{-1}$ &  1 &$1+2q^{1/18}+4q^{1/9}+4q^{1/6}+6q^{2/9}+8q^{5/18}+12q^{1/3}+16q^{7/18}+20q^{4/9}+28q^{1/2}+ \dots$ & \checkmark&
\\
20C &20A &$2^{2}5^{2}20^{2}1^{-2}4^{-2}10^{-2}$ &  -2 &$2+4q^{1/10}+4q^{3/20}+4q^{1/5}+4q^{1/4}+8q^{3/10}+20q^{7/20}+12q^{2/5}+8q^{9/20}+28q^{1/2}+ \dots$& \checkmark&
\\
20D &20A & $1^{2}20^{2}4^{-2}5^{-2}$ & 2 &$2+4q^{1/10}+4q^{3/20}+4q^{1/5}+4q^{1/4}+8q^{3/10}+20q^{7/20}+12q^{2/5}+8q^{9/20}+28q^{1/2}+ \dots$ & \checkmark&
\\
 20E &20B &$4^{1}20^{1}$ &  0 &$0+2q^{1/80}+2q^{1/16}+2q^{9/80}+4q^{13/80}+4q^{17/80}+8q^{21/80}+10q^{5/16}+12q^{29/80}+ \dots$&0&\checkmark
\\
20F &20C &  $2^{2}5^{1}20^{1}1^{-1}4^{-1}$ & -1 &$1+2q^{1/20}+2q^{1/10}+4q^{3/20}+6q^{1/5}+6q^{1/4}+8q^{3/10}+12q^{7/20}+14q^{2/5}+18q^{9/20}+ \dots$& \checkmark&
\\
20G &20C &$1^{1}2^{1}10^{1}20^{1}4^{-1}5^{-1}$ &  1 &$1+2q^{1/20}+2q^{1/10}+4q^{3/20}+6q^{1/5}+6q^{1/4}+8q^{3/10}+12q^{7/20}+14q^{2/5}+18q^{9/20}+ \dots$ & \checkmark&
\\

 21A & 21A &$ 1^{2}21^{2}3^{-2}7^{-2}$ &  2 & $2+4q^{2/21}+4q^{1/7}+8q^{5/21}+8q^{2/7}+4q^{1/3}+20q^{8/21}+24q^{3/7}+8q^{10/21}+32q^{11/21}+ \dots$& \checkmark&
\\
 42A & 21A & $2^{2}3^{2}7^{2}42^{2}1^{-2}6^{-2}14^{-2}21^{-2}$ & -2 & $2+4q^{2/21}+4q^{1/7}+8q^{5/21}+8q^{2/7}+4q^{1/3}+20q^{8/21}+24q^{3/7}+8q^{10/21}+32q^{11/21}+ \dots$& \checkmark&
\\
 21B & 21B & $7^{1}21^{1}1^{-1}3^{-1}$ &  -1 & $1+2q^{1/21}+4q^{2/21}+2q^{1/7}+2q^{4/21}+8q^{5/21}+4q^{2/7}+10q^{1/3}+16q^{8/21}+18q^{3/7}+ \dots$& \checkmark&
\\
 42B & 21B & $1^{1}3^{1}14^{1}42^{1}2^{-1}6^{-1}7^{-1}21^{-1}$ & 1 & $1+2q^{1/21}+4q^{2/21}+2q^{1/7}+2q^{4/21}+8q^{5/21}+4q^{2/7}+10q^{1/3}+16q^{8/21}+18q^{3/7}+ \dots$ & \checkmark&
\\

 21C & 21C & $3^{1}21^{1}$ &  0 & $0+2q^{1/63}+2q^{4/63}+2q^{1/9}+4q^{10/63}+4q^{13/63}+6q^{16/63}+8q^{19/63}+12q^{22/63}++ \dots$&0&\checkmark
\\
 42C & 21C & $6^{1}42^{1}3^{-1}21^{-1}$ &  0 & $0+2q^{1/63}+2q^{4/63}+2q^{1/9}+4q^{10/63}+4q^{13/63}+6q^{16/63}+8q^{19/63}+12q^{22/63}++ \dots$&0&\checkmark
\\
 22BC & 22A & $2^{1}22^{1}$ &  0 & $0+2q^{1/44}+2q^{3/44}+2q^{5/44}+4q^{7/44}+4q^{9/44}+6q^{1/4}+8q^{13/44}+10q^{15/44}+12q^{17/44}+ \dots$&0&\checkmark
\\
 23AB & 23AB & $1^{1}23^{1}$ &  1 & $1+2q^{1/23}+2q^{2/23}+2q^{3/23}+4q^{4/23}+4q^{5/23}+6q^{6/23}+8q^{7/23}+10q^{8/23}+12q^{9/23}+ \dots$ & \checkmark&
\\
 46AB & 23AB & $2^{1}46^{1}1^{-1}23^{-1}$ &  -1 & $1+2q^{1/23}+2q^{2/23}+2q^{3/23}+4q^{4/23}+4q^{5/23}+6q^{6/23}+8q^{7/23}+10q^{8/23}+12q^{9/23}+ \dots$& \checkmark&
\\
 24C & 24A & $2^{2}24^{2}6^{-2}8^{-2}$ &  0 & $0+4q^{1/16}+4q^{5/48}+4q^{3/16}+8q^{11/48}+8q^{5/16}+16q^{17/48}+8q^{19/48}+20q^{7/16}+ \dots$&0&\checkmark
\\
 24D & 24B &$2^{1}3^{2}4^{1}24^{2}1^{-2}6^{-1}8^{-2}12^{-1}$ &  -2 & $2+2q^{1/24}+2q^{1/8}+4q^{1/6}+8q^{1/4}+12q^{7/24}+4q^{1/3}+14q^{3/8}+16q^{5/12}+4q^{11/24}+ \dots$& \checkmark&
\\
 24E & 24B &$1^{2}4^{1}6^{1}24^{2}2^{-1}3^{-2}8^{-2}12^{-1}$ &  2 & $2+2q^{1/24}+2q^{1/8}+4q^{1/6}+8q^{1/4}+12q^{7/24}+4q^{1/3}+14q^{3/8}+16q^{5/12}+4q^{11/24}+ \dots$ & \checkmark&
\\
 24F & 24C & $8^{1}24^{1}2^{-1}6^{-1}$ &  0 & $0+2q^{1/48}+2q^{1/16}+4q^{5/48}+2q^{3/16}+8q^{11/48}+4q^{13/48}+4q^{5/16}+12q^{17/48}+ \dots$&0&\checkmark
\\
 14G & 24D &$12^{1}24^{1}4^{-1}8^{-1}$ &  0 & $0+2q^{1/32}+4q^{7/96}+4q^{5/32}+4q^{19/96}+10q^{9/32}+12q^{31/96}+4q^{35/96}+16q^{13/32}+ \dots$&0&\checkmark
\\
 24H & 24E &$2^{1}6^{1}8^{1}24^{1}4^{-1}12^{-1}$ &  0 & $0+2q^{1/48}+2q^{1/16}+4q^{7/48}+6q^{3/16}+4q^{11/48}+4q^{13/48}+8q^{5/16}+12q^{17/48}+ \dots$&0&\checkmark
\\
 24I & 24F &  $2^{1}3^{1}4^{1}24^{1}1^{-1}8^{-1}$ & -1 & $1+2q^{1/24}+2q^{1/12}+2q^{1/8}+2q^{1/6}+4q^{5/24}+6q^{1/4}+8q^{7/24}+10q^{1/3}+10q^{3/8}+ \dots$ & \checkmark&
\\
 24J & 24F & $1^{1}4^{1}6^{1}24^{1}3^{-1}8^{-1}$ & 1 & $1+2q^{1/24}+2q^{1/12}+2q^{1/8}+2q^{1/6}+4q^{5/24}+6q^{1/4}+8q^{7/24}+10q^{1/3}+10q^{3/8}+ \dots$& \checkmark&
\\
 52A & 26A & $2^{1}52^{1}4^{-1}26^{-1}$ &  0 & $0+2q^{3/104}+2q^{7/104}+2q^{11/104}+2q^{15/104}+4q^{19/104}+4q^{23/104}+4q^{27/104}+8q^{31/104}+ \dots$&0&\checkmark
\\
 28B & 28A & $1^{1}4^{1}7^{1}28^{1}2^{-1}14^{-1}$ & 1 & $1+2q^{1/28}+2q^{1/14}+2q^{1/7}+4q^{5/28}+4q^{3/14}+2q^{1/4}+6q^{2/7}+10q^{9/28}+8q^{5/14}+ \dots$& \checkmark&
\\
 28C & 28A & $4^{1}28^{1}1^{-1}7^{-1}$ &  -1 & $1+2q^{1/28}+2q^{1/14}+2q^{1/7}+4q^{5/28}+4q^{3/14}+2q^{1/4}+6q^{2/7}+10q^{9/28}+8q^{5/14}+ \dots$& \checkmark&
\\
 56AB & 28B & $4^{1}56^{1}8^{-1}28^{-1}$ & 0 & $0+2q^{1/224}+2q^{17/224}+2q^{25/224}+2q^{33/224}+2q^{41/224}+4q^{7/32}+6q^{57/224}+6q^{65/224}+ \dots$&0&\checkmark
\\
 30F & 30A &$1^{1}2^{1}15^{1}30^{1}3^{-1}5^{-1}6^{-1}10^{-1}$ &  1 & $1+2q^{1/15}+4q^{1/10}+4q^{1/6}+2q^{1/5}+10q^{4/15}+8q^{3/10}+2q^{1/3}+8q^{11/30}+12q^{2/5}+ \dots$& \checkmark&
\\
 30G & 30A &$2^{2}3^{1}5^{1}30^{2}1^{-1}6^{-2}10^{-2}15^{-1}$ &  -1 & $1+2q^{1/15}+4q^{1/10}+4q^{1/6}+2q^{1/5}+10q^{4/15}+8q^{3/10}+2q^{1/3}+8q^{11/30}+12q^{2/5}+ \dots$& \checkmark&
\\
 60A & 30B & $2^{1}10^{1}12^{1}60^{1}4^{-1}6^{-1}20^{-1}30^{-1}$ & 0 & $0+2q^{1/40}+4q^{11/120}+2q^{1/8}+4q^{23/120}+6q^{9/40}+8q^{7/24}+8q^{13/40}+4q^{43/120}+ \dots$&0&\checkmark
\\
 60B & 30C & $6^{1}60^{1}12^{-1}30^{-1}$ & 0 & $0+2q^{1/360}+2q^{5/72}+2q^{49/360}+4q^{61/360}+4q^{73/360}+4q^{17/72}+4q^{97/360}+4q^{109/360}+ \dots$&0&\checkmark
\\

 30H & 30D &$1^{1}6^{1}10^{1}15^{1}3^{-1}5^{-1}$ &  1 & $1+2q^{1/30}+2q^{1/10}+2q^{2/15}+2q^{1/6}+4q^{1/5}+4q^{7/30}+4q^{4/15}+6q^{3/10}+8q^{1/3}+ \dots$& \checkmark&
\\
 30I & 30D &$2^{1}3^{1}5^{1}30^{1}1^{-1}15^{-1}$ &  -1 & $1+2q^{1/30}+2q^{1/10}+2q^{2/15}+2q^{1/6}+4q^{1/5}+4q^{7/30}+4q^{4/15}+6q^{3/10}+8q^{1/3}+ \dots$& \checkmark&
\\
 30J & 30E &$2^{1}3^{1}5^{1}30^{1}6^{-1}10^{-1}$ &  0 & $0+2q^{1/30}+2q^{1/15}+2q^{1/10}+2q^{2/15}+2q^{1/6}+2q^{1/5}+4q^{7/30}+6q^{4/15}+6q^{3/10}+ \dots$&0&\checkmark
\\
 30K & 30E & $2^{1}30^{1}3^{-1}5^{-1}$ & 0 & $0+2q^{1/30}+2q^{1/15}+2q^{1/10}+2q^{2/15}+2q^{1/6}+2q^{1/5}+4q^{7/30}+6q^{4/15}+6q^{3/10}+ \dots$&0&\checkmark
\\
 33A & 33A & $3^{1}33^{1}1^{-1}11^{-1}$ & -1 & $1+2q^{1/33}+2q^{1/11}+2q^{4/33}+4q^{2/11}+4q^{7/33}+6q^{3/11}+8q^{10/33}+2q^{1/3}+10q^{4/11}+ \dots$& \checkmark&
\\

66A &33A & $1^{1}6^{1}11^{1}66^{1}2^{-1}3^{-1}22^{-1}33^{-1}$ & 1 &$1+2q^{1/33}+2q^{1/11}+2q^{4/33}+4q^{2/11}+4q^{7/33}+6q^{3/11}+8q^{10/33}+2q^{1/3}+10q^{4/11} \dots$ & \checkmark&
\\
35A &35A &$1^{1}35^{1}5^{-1}7^{-1}$ &  1 &$1+2q^{1/35}+2q^{4/35}+2q^{1/7}+4q^{6/35}+2q^{1/5}+4q^{8/35}+2q^{9/35}+4q^{2/7}+8q^{11/35}+ \dots$ & \checkmark&
\\
70A &35A &$2^{1}5^{1}7^{1}70^{1}1^{-1}10^{-1}14^{-1}35^{-1}$ &  -1 &$1+2q^{1/35}+2q^{4/35}+2q^{1/7}+4q^{6/35}+2q^{1/5}+4q^{8/35}+2q^{9/35}+4q^{2/7}+8q^{11/35}+ \dots$& \checkmark&
\\
36A &36A &$2^{1}9^{1}36^{1}1^{-1}4^{-1}18^{-1} $&  -1 &$1+2q^{1/18}+2q^{1/12}+2q^{1/9}+2q^{1/6}+4q^{7/36}+2q^{2/9}+2q^{1/4}+4q^{5/18}+8q^{11/36}+ \dots$& \checkmark&
\\
36B &36A & $1^{1}36^{1}4^{-1}9^{-1} $& 1 &$1+2q^{1/18}+2q^{1/12}+2q^{1/9}+2q^{1/6}+4q^{7/36}+2q^{2/9}+2q^{1/4}+4q^{5/18}+8q^{11/36}+ \dots$& \checkmark&
\\
39AB &39AB &$ 1^{1}39^{1}3^{-1}13^{-1} $& 1 &$1+2q^{2/39}+2q^{1/13}+2q^{5/39}+2q^{2/13}+4q^{8/39}+4q^{3/13}+6q^{11/39}+6q^{4/13}+2q^{1/3}+ \dots$& \checkmark&
\\
78AB &39AB &$2^{1}3^{1}13^{1}78^{1}1^{-1}6^{-1}26^{-1}39^{-1} $&  -1 &$1+2q^{2/39}+2q^{1/13}+2q^{5/39}+2q^{2/13}+4q^{8/39}+4q^{3/13}+6q^{11/39}+6q^{4/13}+2q^{1/3}+ \dots$& \checkmark&
\\
40AB &40A &$2^{1}40^{1}8^{-1}10^{-1} $&  0 &$0+2q^{3/80}+2q^{1/16}+2q^{7/80}+4q^{13/80}+2q^{3/16}+4q^{19/80}+4q^{21/80}+6q^{23/80}+4q^{5/16}+ \dots$&0&\checkmark
\\
84A &42A &$4^{1}6^{1}14^{1}84^{1}2^{-1}12^{-1}28^{-1}42^{-1} $&  0 &$0+2q^{1/168}+2q^{3/56}+2q^{1/8}+2q^{25/168}+4q^{11/56}+4q^{37/168}+4q^{15/56}+6q^{7/24}+ \dots$&0&\checkmark
\\
60C &60A & $\begin{matrix} 1^{1}4^{1}6^{1}10^{1}15^{1}60^{1}2^{-1} \qquad \\ \qquad 3^{-1}5^{-1}12^{-1}20^{-1}30^{-1} \end{matrix}$ &  1 &$1+2q^{1/60}+2q^{1/10}+2q^{3/20}+2q^{1/6}+2q^{1/5}+4q^{13/60}+2q^{1/4}+4q^{4/15}+4q^{3/10}+ \dots$& \checkmark&
\\
60D &60A &$3^{1}4^{1}5^{1}60^{1}1^{-1}12^{-1}15^{-1}20^{-1}$ & -1 &$1+2q^{1/60}+2q^{1/10}+2q^{3/20}+2q^{1/6}+2q^{1/5}+4q^{13/60}+2q^{1/4}+4q^{4/15}+4q^{3/10}+ \dots$& \checkmark&
 \\
 \end{longtable}
\end{landscape}

\printbibliography

 \end{document}